\documentclass{amsart}
\usepackage{amsmath,amsfonts,amssymb,amsthm,mathtools}
\usepackage{amscd}
\usepackage{bbm}
\usepackage{enumerate}
\usepackage{galois}
\usepackage{mathrsfs}
\usepackage{xypic}
\usepackage{geometry}
\usepackage{hyperref}

\usepackage{color}

\theoremstyle{plain}
\newtheorem{Th}{Theorem}[section]

\newtheorem{Lem}[Th]{Lemma}
\newtheorem{Prop}[Th]{Proposition}

\theoremstyle{plain}
\newtheorem*{Main}{Main Theorem}

\numberwithin{equation}{section}

\newcommand{\diff}[2]{\frac{\partial #1}{\partial #2}}

\newcommand{\qa}{\alpha}
\newcommand{\qb}{\beta}
\newcommand{\qd}{\delta}
\newcommand{\qg}{\gamma}

\newcommand{\qt}{\tau}
\newcommand{\qth}{\theta}
\newcommand{\qe}{\varepsilon}
\newcommand{\qz}{\zeta}
\newcommand{\qp}{\partial}

\newcommand{\Qg}{\Gamma}
\newcommand{\ql}{\lambda}

\newcommand{\Qd}{\Delta}
\newcommand{\Qo}{\Omega}

\newcommand{\kk}[1]{\left(#1\right)}

\newcommand{\fk}[2]{\left[#1, #2\right]}

\begin{document}

\title{Linearization of Virasoro symmetries associated with semisimple Frobenius manifolds}

\author{Si-Qi Liu, Zhe Wang, Youjin Zhang}
\keywords{Frobenius manifold, Principal Hierarchy, Bihamiltonian structure, Virasoro symmetry, Loop equation}

\begin{abstract}

For any semisimple Frobenius manifold, we prove that a tau-symmetric bihamiltonian deformation of its Principal Hierarchy admits an infinite family of linearizable Virasoro symmetries if and only if all the central invariants of the corresponding deformation of the bihamiltonian structure are equal to $\frac{1}{24}$. As an important application of this result, we prove that the Dubrovin-Zhang hierarchy associated with the semisimple Frobenius manifold possesses a bihamiltonian structure which can be represented in terms of differential polynomials.
\end{abstract}

\date{\today}

\maketitle
\tableofcontents


\section{Introduction}\label{intro}
The deep relationship between  2D topological field theories (2DTFT) and integrable hierarchies has become one of the central research topics in mathematical physics since the proof of the Witten conjecture\cite{witten1990two} by  Kontsevich \cite{kontsevich1992intersection}, which relates the intersection numbers on the moduli space of stable curves to the well-known Korteweg-de Vries (KdV) hierarchy. In \cite{dubrovin1996geometry}, Dubrovin introduced the notion of Frobenius manifolds as a geometric interpretation of the genus zero part of a 2DTFT, and showed that the genus zero partition function of a 2DTFT is the tau-function of a particular solution of a bihamiltonian integrable hierarchy, which is called the Principal Hierarchy of the Frobenius manifold corresponding to the 2DTFT.

In \cite{dubrovin2001normal} Dubrovin and the third-named author constructed a certain deformation of the Principal Hierarchy associated with a semisimple Frobenius manifold and conjectured that the full genera partition function of the 2DTFT corresponding to the Frobenius manifold is the tau-function of a solution of the deformed integrable hierarchy. This particular deformed  integrable hierarchy is called the topological deformation of the Principal Hierarchy and it is also called the Dubrovin-Zhang (DZ) hierarchy in the literature. It is constructed by performing a quasi-Miura transformation to the Principal Hierarchy and consists of bihamiltonian evolutionary PDEs. In \cite{buryak2012deformations,buryak2012polynomial}, Buryak, Posthuma and Shadrin proved that the DZ hierarchy and its  first Hamiltonian structure can be represented by differential polynomials. However, the polynomiality of the second Hamiltonian structure of the DZ hierarchy remains unproved. In \cite{iglesias2021bi}, Hern{\'a}ndez Iglesias and Shadrin provided some evidences which support the validity of the polynomiality of the second Hamiltonian structure by using the geometry of the moduli space of the stable curves.

%

In this paper, we prove the polynomiality of the second Hamiltonian structure of the DZ hierarchy.The approach of our proof consists of the following two steps:

\textbf{Step 1.} To prove that any polynomial tau-symmetric bihamiltonian deformation of the Principal Hierarchy possesses an infinite family of Virasoro symmetries.

\textbf{Step 2.} To prove that the Virasoro symmetries constructed in Step 1. can be linearized (in a sense that will be explained later) if the central invariants of the bihamiltonian structure of the deformed integrable hierarchy are all equal to $\frac{1}{24}$.

In \cite{dubrovin2001normal}, it is proved that the integrable hierarchy obtained in these two steps is unique and it is exactly the DZ hierarchy. This implies that the bihamiltonian structure of the DZ hierarchy is polynomial. We have finished the first step in \cite{variationalII,liu2020super,liu2021variational} by developing the theory of super tau-covers of bihamiltonian integrable hierarchies and variational bihamiltonian cohomologies. In particular we have the following theorem. 
\begin{Th}[\cite{variationalII}]
\label{AT}
For the Principal Hierarchy associated with a semisimple Frobenius manifold and any of its tau-symmetric bihamiltonian deformations, there exists a unique deformation of its Virasoro symmetries such that they are symmetries of the deformed integrable hierarchy. Moreover, the action of the Virasoro symmetries on the tau-function $Z$ of the deformed integrable hierarchy can be represented in the form 
\begin{equation*}
\diff{Z}{s_m} = L_mZ+O_mZ,\quad m\geq -1,
\end{equation*}
where $L_m$ are the Virasoro operators constructed in \cite{dubrovin1999frobenius} and $O_m$ are some differential polynomials, and the flows $\diff{}{s_m}$ satisfy the Virasoro commutation relations
\[
\fk{\diff{}{s_k}}{\diff{}{ s_l}} = (l-k)\diff{}{ s_{k+l}},\quad k,l\geq -1.
\]
\end{Th}

In this paper, we proceed to finish the second step. Let us first explain this step in a more explicit way. It is proved in \cite{dubrovin2018bihamiltonian} that the tau-structure of the deformed Principal Hierarchy changes if  we choose different representatives for the deformation of the bihamiltonian structure under Miura type transformations. In particular, the differential polynomials $O_m$ described in the above theorem change accordingly after performing a Miura type transformation to the deformed integrable hierarchy. So by the term \textit{to linearize the Virasoro symmetries} we mean to find a suitable Miura type transformation such that after performing this transformation, the differential polynomials $O_m$ become zero (see Sect.\,\ref{line} for a detailed description). We prove the following theorem.
\begin{Main}
For a given tau-symmetric bihamiltonian deformation for the Principal Hierarchy associated with a semisimple Frobenius manifold, the Virasoro symmetries of the deformed integrable hierarchy can be linearized  if and only if the central invariants of the deformed bihamiltonian structure are all equal to $\frac{1}{24}$, i.e., in this case, by choosing a suitable representative of the bihamiltonian structure, the action of the Virasoro symmetries on the tau-function $Z$ of the deformed integrable hierarchy can be represented in the form 
\begin{equation*}
\diff{Z}{s_m} = L_mZ,\quad m\geq -1.
\end{equation*}
\end{Main}
As an immediate corollary, we have the following polynomiality theorem for the DZ hierarchy.
\begin{Th}
The bihamiltonian structure of the DZ hierarchy associated with a semisimple Frobenius manifold can be represented by  differential polynomials.

\end{Th}

This paper is organized as follows. In Sect.\,\ref{loop}, we recall some basic facts about the Frobenius manifolds and their loop equations. In Sect.\,\ref{line} we give the prove of the main theorem. In Sect.\,\ref{conc} we make some concluding remarks.


\section{Frobenius manifold and its loop equation}\label{loop}
The notion of Frobenius manifolds \cite{dubrovin1996geometry,dubrovin2001normal} gives a coordinate-free description of the Witten-Dijkgraaf-Verlinde-Verlinde (WDVV) associativity equation \cite{dijkgraaf1991topological,witten1990structure}. An $n$-dimensional Frobenius manifold $M$ is an analytic manifold equipped with a flat (pseudo-)Riemannian metric $\eta$ and a $(1,2)$-tensor $c$ which define Frobenius algebra structures on the tangent spaces $T_pM$ for any $p\in M$ with flat unit vector field. Furthermore these data satisfy certain homogeneous conditions.  Let $v^1,\cdots, v^n$ be flat coordinates of $M$ with $\diff{}{v^1}$ being the unit vector field, then the local structure of $M$ can be described by a solution $F(v^1,\cdots, v^n)$ of the following WDVV equation:
\begin{equation*}
\qp_\qa\qp_\qb\qp_\ql F\eta^{\ql\mu}\qp_\mu\qp_\qg\qp_\qd F = \qp_\qd\qp_\qb\qp_\ql F\eta^{\ql\mu}\qp_\mu\qp_\qg\qp_\qa F,
\end{equation*}
where $(\eta^{\ql\mu}) = (\eta_{\ql\mu})^{-1}$ and $\qp_\qa = \diff{}{v^{\qa}}$. Here and henceforth the summation over repeated upper and lower Greek indices is assumed. This function $F$ is called the potential of the Frobenius manifold $M$, and it is related to the tensors $\eta$ and $c$ by
\[
\eta_{\qa\qb} = \qp_1\qp_\qa\qp_\qb F,\quad c_{\qa\qb}^\qg = \eta^{\qg\ql}\qp_\ql\qp_\qa\qp_\qb F.
\]
The potential is required to be quasi-homogeneous with respect to an Euler vector field of the form 
\begin{equation}
\label{AJ}
E = \sum_{\qa=1}^n \left(\left(1-\frac{d}{2}-\mu_\qa\right)v^\qa+r_\qa\right)\qp_\qa,
\end{equation}
i.e.,
\[ E(F)=(3-d)F+\frac12 A_{\qa\qb} v^\qa v^\qb+B_\qa v^\qa+C.\]
Here $\mu_\qa$, $r_\qa$, $A_{\qa\qb}$, $B_\qa$ and $C$ are some constants and the constant $d$ is called the charge of $M$. 

The Frobenius algebra structure on $TM$ yields a deformed flat connection
\begin{equation}
\label{AB}
\tilde\nabla_X Y = \nabla_X Y+zX\cdot Y,\quad \forall\, X,Y\in\Qg(TM),\quad z\in\mathbb C,
\end{equation}
here $\nabla$ is the Levi-Civita connection of $\eta$ and $X\cdot Y = c(X,Y)$. It can be extended to be an affine connection on $M\times \mathbb C^*$ by regarding $z$ as the coordinate of $\mathbb C^*$ and by requiring
\[
\tilde\nabla_{\qp_z}X = \qp_z X+E\cdot X-\frac 1z \mu(X),\quad \tilde\nabla_{\qp_z}\qp_z = \tilde\nabla_{X}\qp_z = 0,
\]
where $X$ is a vector field on $M\times \mathbb C^*$ whose $\qp_z$ component is zero and $\mu = \mathrm{diag}(\mu_1,\cdots,\mu_n)$ is a diagonal matrix regarding as a section of the vector bundle $\mathrm{End}(TM)$.

A system of flat coordinates $\tilde v_1,\cdots,\tilde v_n$ of the deformed flat connection can be chosen to have the following form:
\[
(\tilde v^1(v,z),\cdots,\tilde v^n(v,z)) = (\qth_1(v,z),\cdots,\qth_n(v,z))z^\mu z^R,
\]
here $R$ is a certain constant matrix. The functions $\qth_\qa(v,z)$ are analytic at $z=0$ and has the expansion $\qth_\qa(v,z) = \sum_{p\geq 0}\qth_{\qa,p}(v)z^p$. The coefficients $\qth_{\qa,p}$ satisfy the following recursion relations:
\[
\qth_{\qa,0} = \eta_{\qa\qb}v^\qb,\quad\qp_\qb\qp_\qg \qth_{\qa,p+1} = c^\ql_{\qb\qg}\qp_\ql \qth_{\qa,p},\quad p\geq 0.
\]
They also satisfy the quasi-homogeneous condition
\[
E(\qp_\qb \qth_{\qa,p}) = (p+\mu_\qa+\mu_\qb)\qp_\qb \qth_{\qa,p}+\sum_{k=1}^p(R_k)^\qg_\qa\qp_\qb \qth_{\qg,p-k}.
\]

The Principal Hierarchy associated with $M$ is defined to be the following mutually commuting flows:
\begin{equation}
\label{AK}
\diff{v^\ql}{t^{\qa,p}} = \eta^{\ql\qg}\qp_x(\qp_\qg \qth_{\qa,p+1}),\quad \qa = 1,\cdots,n,\quad p\geq 0.
\end{equation}
Since $\diff{v^\ql}{t^{1,0}} = \qp_xv^\ql$, we will identify the time variable $t^{1,0}$ with the spacial variable $x$. We will also use the notation $v^{\ql,k} = \qp_x^kv^\ql$ in what follows. The Principal Hierarchy \eqref{AK} is bihamiltonian with respect to the following Poisson pencil:
\begin{align}
\label{AM}
\{v^\qa(x),v^\qb(y)\}_1^{[0]} &= \eta^{\qa\qb}\qd'(x-y),\\
\label{AN}
\{v^\qa(x),v^\qb(y)\}_2^{[0]} &= g^{\qa\qb}(v(x))\qd'(x-y)+\Qg^{\qa\qb}_\qg v^\qg_x(x)\qd(x-y).
\end{align}
Here the functions $g^{\qa\qb}$ and $\Qg^{\qa\qb}_\qg$ are given by
\[
g^{\qa\qb} = E^\qe c^{\qa\qb}_\qe,\quad \Qg^{\qa\qb}_\qg = \left(\frac{1}{2}-\mu_\qb\right)c^{\qa\qb}_\qg
\]
with $c^{\qa\qb}_\qg = \eta^{\qa\ql}c^\qb_{\ql\qg}$. In terms of the first Hamiltonian structure, the flows of the  Principal Hierarchy can be represented as the following Hamiltonian systems:
\[
\diff{v^\ql}{t^{\qa,p}} = \{v^\ql,\int \qth_{\qa,p+1}\}_1^{[0]}.
\]
It is also tau-symmetric, i.e., the densities of the Hamiltonians satisfy the symmetry property 
\[
\diff{\qth_{\qa,p}}{t^{\qb,q}} = \diff{\qth_{\qb,q}}{t^{\qa,p}}.
\]
The tau-symmetric property of the Principal Hierarchy enables us to define a set of functions $\Qo_{\qa,p;\qb,q}^{[0]}(v)$ for $\qa,\qb = 1,\cdots,n$ and $p,q\geq 0$ such that the following identities hold true:
\begin{align*}
&\Qo^{[0]}_{\qa,p;1,0} = \qth_{\qa,p},\quad \Qo^{[0]}_{\qa,p;\qb,0} = \qp_\qb \qth_{\qa,p+1};\\
&\Qo_{\qa,p;\qb,q}^{[0]} = \Qo_{\qb,q;\qa,p}^{[0]},\quad \diff{\Qo_{\qa,p;\qb,q}^{[0]}}{t^{\ql,k}}= \diff{\Qo_{\ql,k;\qb,q}^{[0]}}{t^{\qa,p}}.
\end{align*}
Let us introduce the tau-cover of the Principal Hierarchy as follows:
\[
\diff{ \mathcal F_0}{t^{\qa,p}} = f_{\qa,p}^{[0]},\quad \diff{f_{\qb,q}^{[0]}}{t^{\qa,p}} = \Qo_{\qa,p;\qb,q}^{[0]},\quad \diff{v^\qb}{t^{\qa,p}} = \eta^{\qb\ql}\qp_x\Qo_{\ql,0;\qa,p}^{[0]}.
\]
The unknown function $ \mathcal F_0$ is called the genus zero free energy and the functions $f_{\qa,p}^{[0]}$ are called the genus zero one-point functions.

In \cite{dubrovin1999frobenius}, an infinite set of Virasoro symmetries of the tau-cover of the Principal Hierarchy was constructed. These symmetries are closely related to the differential operators
\[L_m=a_m^{\qa,p;\qb,q}\qe^2\frac{\qp^2}{\qp t^{\qa,p}\qp t^{\qb,q}}+{b}_{m;\qa,p}^{\qb,q} t^{\qa,p}\frac{\qp}{\qp t^{\qb,q}}+c_{m;\qa,p;\qb,q}\frac{1}{\qe^2}t^{\qa,p} t^{\qb,q}+\frac{1}{4}\qd_{m,0}\mathrm{tr}\left(\frac{1}{4}-\mu^2\right).
\]
where $m\geq -1$ and $a_m^{\qa,p;\qb,q}$, ${b}_{m;\qa,p}^{\qb,q}$, $c_{m;\qa,p;\qb,q}$ are some constants determined from the monodromy data of the Frobenius manifold $M$. They can be represented in the form
\begin{align}
\label{AC}
\diff{ \mathcal F_0}{s_m} &= a_m^{\qa,p;\qb,q}f_{\qa,p}^{[0]}f_{\qb,q}^{[0]}+{b}_{m;\qa,p}^{\qb,q}t^{\qa,p}f_{\qb,q}^{[0]}+c_{m;\qa,p;\qb,q} t^{\qa,p} t^{\qb,q},\\
\label{AD}
\diff{f_{\qa,p}^{[0]}}{s_m} &= \diff{}{t_{\qa,p}}\diff{ \mathcal F_0}{s_m},\quad \diff{v^\qa}{s_m} = \eta^{\qa\qb}\diff{}{t^{1,0}}\diff{f_{\qb,0}^{[0]}}{s_m}.
\end{align}
These flows satisfy the Virasoro commutation relations
\[
\fk{\diff{}{s_k}}{\diff{}{ s_l}} = (l-k)\diff{}{ s_{k+l}},\quad k,l\geq -1.
\]

In \cite{dubrovin2001normal}, a deformation of the Principal Hierarchy associated with a semisimple Frobenius manifold $M$ is constructed by requiring that the actions of the Virasoro symmetries on the tau-function 
\[
Z = \exp\biggl(\frac{1}{\qe^2}\mathcal F_0+\sum_{g\geq 1}\qe^{2g-2} F_g\bigl(v,v_x,v_{xx},\cdots,v^{(m_g)}\bigr)\biggr)
\]
of the deformed integrable hierarchy are represented by
\begin{equation}
\label{AE}
\diff{Z}{s_m} = L_m Z,\quad m\geq -1.
\end{equation}
This procedure is called the linearization of the Virasoro symmetries, and the functions $ F_g$ are called the genus $g$ free energy of the deformed integrable hierarchy, which are determined by the equations
\begin{align}
\label{AF}
\diff{\Qd F}{s_m}=&\,a_m^{\qa,p;\qb,q}\kk{2\diff{ \mathcal F_0}{t^{\qa,p}}\diff{\Qd F}{t^{\qb,q}}+\qe^2\diff{\Qd F}{t^{\qa,p}}\diff{\Qd F}{t^{\qb,q}}+\frac{\qp^2( \mathcal F_0+\qe^2\Qd F)}{\qp t^{\qa,p}\qp t^{\qb,q}}}\\
\notag
&+{b}_{m;\qa,p}^{\qb,q} t^{\qa,p}\frac{\qp\Qd F}{\qp t^{\qb,q}}+\frac{1}{4}\qd_{m,0}\mathrm{tr}\left(\frac{1}{4}-\mu^2\right),\quad m\geq -1,
\end{align}
here we denote
\[\Qd F = \sum_{g\geq 1}\qe^{2g-2} F_g\bigl(v,v_x,v_{xx},\cdots,v^{(m_g)}\bigr).
\]

Denote by $\mathcal D_m$ the following derivation:
\begin{equation}
\label{AX}
\mathcal D_m = 2a_m^{\qa,p;\qb,q}\diff{ \mathcal F_0}{t^{\qa,p}}\diff{}{t^{\qb,q}}+{b}_{m;\qa,p}^{\qb,q} t^{\qa,p}\frac{\qp}{\qp t^{\qb,q}}-\diff{}{s_m},\quad m\geq -1,
\end{equation}
then the equations in \eqref{AF} can be rewritten as
\begin{align}
\label{AG}
\mathcal D_m\Qd F =&\, -a_m^{\qa,p;\qb,q}\frac{\qp^2 \mathcal F_0}{\qp t^{\qa,p}\qp t^{\qb,q}}-\qe^2a_m^{\qa,p;\qb,q}\kk{\diff{\Qd F}{t^{\qa,p}}\diff{\Qd F}{t^{\qb,q}}+\frac{\qp^2\Qd F}{\qp t^{\qa,p}\qp t^{\qb,q}}}\\
\notag
&-\frac{1}{4}\qd_{m,0}\mathrm{tr}\left(\frac{1}{4}-\mu^2\right),\quad m\geq -1.
\end{align}
By comparing the coefficients of $\qe^{2g-2}$ for $g\geq 1$, we see that we can represent the above equations as
\begin{align*}
\mathcal D_m F_1 =&\, -a_m^{\qa,p;\qb,q}\frac{\qp^2 \mathcal F_0}{\qp t^{\qa,p}t^{\qb,q}}-\frac{1}{4}\qd_{m,0}\mathrm{tr}\left(\frac{1}{4}-\mu^2\right),\\
\mathcal D_m F_g=&\, -a_m^{\qa,p;\qb,q}\kk{\sum_{k=1}^{g-1}\diff{ F_k}{t^{\qa,p}}\diff{ F_{g-k}}{t^{\qb,q}}+\frac{\qp^2 F_{g-1}}{\qp t^{\qa,p}\qp t^{\qb,q}}},\quad g\geq 2.
\end{align*}

To solve equations \eqref{AG} for all $m\geq -1$, the authors of \cite{dubrovin2001normal} reorganized them into the form
\begin{align}
\label{AH}
\mathcal D(\ql)\Qd F =&\, -\sum_{m\geq -1}\frac{1}{\ql^{m+2}}a_m^{\qa,p;\qb,q}\kk{\frac{\qp^2 \mathcal F_0}{\qp t^{\qa,p}\qp t^{\qb,q}}+\qe^2\diff{\Qd F}{t^{\qa,p}}\diff{\Qd F}{t^{\qb,q}}+\qe^2\frac{\qp^2\Qd F}{\qp t^{\qa,p}\qp t^{\qb,q}}}\\
\notag
&-\frac{1}{4\ql^2}\mathrm{tr}\left(\frac{1}{4}-\mu^2\right),
\end{align}
where $\mathcal D(\ql) = \sum_{m\geq -1}\mathcal D_m/\ql^{m+2}$. An important observation made in \cite{dubrovin2001normal} is that the equation \eqref{AH} can be rewritten as an equation for $\Qd F$ on the jet space, which has the following explicit form (Lemma 4.2.6 of \cite{dubrovin2001normal}):
\begin{align}
\label{AI}
&\sum_{r\geq 0}\diff{\Qd F}{v^{\qg,r}}\qp_x^r\kk{\frac{1}{E-\ql}}^\qg+\sum_{r\geq 1}\diff{\Qd F}{v^{\qg,r}}\sum_{k=1}^r\binom{r}{k}\qp_x^{k-1}\qp_1p_\qa G^{\qa\qb}\qp_x^{r-k+1}\qp^\qg p_\qb\\
\notag
=&\,\frac 12 (\qp_\ql p_\qa)*(\qp_\ql p_\qb)G^{\qa\qb}-\frac{1}{4\ql^2}\mathrm{tr}\left(\frac{1}{4}-\mu^2\right)\\
\notag
&+\frac{\qe^2}{2}\sum\kk{\diff{\Qd F}{v^{\qg,k}}\diff{\Qd F}{v^{\rho,l}}+\frac{\qp^2\Qd F}{\qp v^{\qg,k}\qp v^{\rho,l}}}\qp_x^{k+1}\qp^\qg p_\qa G^{\qa\qb}\qp_x^{l+1}\qp^\rho p_\qb\\
\notag
&+\frac{\qe^2}{2}\sum\diff{\Qd F}{v^{\qg,k}}\qp_x^{k+1}\left[\nabla\diff{p_\qa}{\ql}\cdot \nabla\diff{p_\qb}{\ql}\cdot v_x\right]^\qg G^{\qa\qb},
\end{align}
here $E$ is the Euler vector field \eqref{AJ}, $p_\qa = p_\qa(v;\ql)$ are the periods of the Frobenius manifold $M$,  $G^{\qa\qb}$ is a constant matrix and $*$ is a multiplication defined on the space of conserved quantities of the Principal Hierarchy. One may refer to Sect.\,4 of \cite{dubrovin2001normal} for details. The equation \eqref{AI} is called the loop equation of the Frobenius manifold $M$ and it plays a central role in the construction of the DZ hierarchy.

It is proved that the loop equation \eqref{AI} has a unique solution up to addition of constants if the Frobenius manifold $M$ is semisimple. Recall that a Frobenius manifold $M$ is called semisimple if the Frobenius algebra structure on $T_pM$ is semisimple for generic $p\in M$. For a semisimple Frobenius manifold, there exists local coordinates $u^1,\cdots u^n$ such that $\diff{}{u^1},\cdots,\diff{}{u^n}$ form a basis of idempotents of the Frobenius algebra on $TM$, i.e.,
\[
c\left(\diff{}{u^i},\diff{}{u^j}\right) = \qd_{i,j}\diff{}{u^i},\quad i,j = 1,\cdots,n.
\]
Such coordinates are unique up to permutations and are called the canonical coordinates of $M$. The following theorem is proved in \cite{dubrovin2001normal}.
\begin{Th}[\cite{dubrovin2001normal}]
\label{AL}
If the Frobenius manifold $M$ is semisimple, then the loop equation \eqref{AI} has a unique, up to the addition of constants, solution
\[
\Qd F = \sum_{g\geq 1}\qe^{2g-2} F_g\bigl(v,v_x,v_{xx},\cdots,v^{(3g-2)}\bigr).
\]
In particular we have
\begin{align*}
 F_1 &= \log\frac{\qt_I(u)}{J^{1/24}(u)}+\frac{1}{24}\sum_{i=1}^n\log u^i_x,\\
 F_g&\in C^\infty(v)\left[v_x,v_{xx},\cdots,v^{(3g-2)}\right]\left[\frac{1}{u^1_x\cdots u^n_x}\right],\quad g\geq 2,
\end{align*}
where $u^1,\cdots,u^n$ are canonical coordinates of $M$, $\qt_I(u)$ is the isomonodromic tau-function of $M$ and $J(u)$ is the Jacobian $\det\kk{\diff{v^\qa}{u^i}}$.
\end{Th}
By using the above theorem, the DZ hierarchy associated with a semisimple Frobenius manifold $M$ is obtained by performing the quasi-Miura transformation
\[
v^\qa\mapsto w^\qa = v^\qa+\eta^{\qa\qb}\qe^2\frac{\qp^2\Qd F}{\qp t^{\qb,0}\qp t^{1,0}},\quad \qa = 1,\cdots, n
\]
to the Principal Hierarchy \eqref{AK} and $w^\qa$ are called the normal coordinates of the DZ hierarchy. The tau-cover of the DZ hierarchy can be obtained by the same quasi-Miura transformation from that of the Principal Hierarchy, it can be represented as
\[
\qe\diff{\mathcal F}{t^{\qa,p}} = f_{\qa,p},\quad \qe\diff{f_{\qb,q}}{t^{\qa,p}} = \Qo_{\qa,p;\qb,q},\quad \diff{w^\qb}{t^{\qa,p}} = \eta^{\qb\ql}\qp_x\Qo_{\ql,0;\qa,p},
\]
where the functions $\Qo_{\qa,p;\qb,q}$ is given by
\[
\Qo_{\qa,p;\qb,q} = \Qo_{\qa,p;\qb,q}^{[0]}+\qe^2\frac{\qp^2\Qd F}{\qp t^{\qa,p}\qp t^{\qb,q}}.
\]
We call $\mathcal F$, $f_{\qa,p}$ and $\Qo_{\qa,p;\qb,q}$ the free energy, the one-point functions and the two-point functions of the DZ hierarchy respectively.  Moreover the tau-cover of the DZ hierarchy admits linearized Virasoro symmetries
\[
\diff{\mathcal F}{s_m} = \exp(- \mathcal F)L_m(\exp(\mathcal F)),\quad\diff{f_{\qa,p}}{s_m} = \qe\diff{}{t^{\qa,p}}\diff{\mathcal F}{s_m},\quad \diff{w^\qa}{s_m} = \qe\eta^{\qa\qb}\diff{}{t^{1,0}}\diff{f_{\qb,0}}{s_m},
\]
where $m\geq -1$.

\section{Linearization of Virasoro symmetries}
\label{line}
\subsection{Bihamiltonian structure and quasi-triviality}
Throughout this section, we fix a semisimple Frobenius manifold $M$ of dimension $n$. Let us consider a deformation of the bihamiltonian structure \eqref{AM},\eqref{AN} of the Principal Hierarchy of $M$ with constant central invariants. It follows from \cite{dubrovin2018bihamiltonian} that the deformed bihamiltonian structure uniquely determines a deformation of the Principal Hierarchy which is tau-symmetric. In another word, the deformed bihamiltonian structure determines a deformation of the tau-cover of the Principal Hierarchy, which can be represented by
\begin{equation*}\qe\diff{\mathcal F}{t^{\qa,p}} = f_{\qa,p},\quad\qe\diff{f_{\qb,q}}{t^{\qa,p}} = \Qo_{\qa,p;\qb,q},\quad \diff{w^\qb}{t^{\qa,p}} = \eta^{\qb\ql}\qp_x\Qo_{\ql,0;\qa,p}.
\end{equation*}
Here the two-point functions $\Qo_{\qa,p;\qb,q}$ are differential polynomials in $w^1,\cdots, w^n$ whose leading terms are $\Qo_{\qa,p;\qb,q}^{[0]}$ after we replace each $w^\ql$ by $v^\ql$. The unknown functions $w^\qa$ are called the normal coordinates of the deformed integrable hierarchy. In terms of the normal coordinates, we can represent the deformed bihamiltonian structure in the form
\[
\{w^\qa(x),w^\qb(y)\}_a = \sum_{g\geq 0}\sum_{k=0}^{2g+1}\qe^{2g}P_{a;g,k}^{\qa\qb}(x)\qd^{(2g+1-k)}(x-y),\quad a = 1,2.
\]
Here $P_{a;g,k}^{\qa\qb}$ are differential polynomials in $w^1,\cdots, w^n$ of differential degree $k$. Recall that the ring of differential polynomials is graded with respect to the differential degree $\deg_x$ defined by $\deg_x w^{\qa,r} = r$. It was proved in \cite{DLZ-1} that such a deformation is quasi-trivial, i.e., there exists a unique quasi-Miura transformation
\begin{equation}
\label{AQ}
v^\qa\mapsto w^\qa = v^\qa+\sum_{g\geq 1}\qe^{2g}Q^\qa_{g}(v_x,v_{xx},\cdots,v^{(3g)}),\quad Q_g^\qa\in \mathcal S_{2g}
\end{equation}
such that the Poisson bracket $\{-,-\}_a$ is obtained from $\{-,-\}_a^{[0]}$ by the above transformation. Here we denote
\[
\mathcal S =  C^\infty(v)\left[v_x,v_{xx},\cdots\right]\left[\frac{1}{u^1_x\cdots u^n_x}\right].
\]
Note that $\mathcal S$ is also graded by the differential degree with $\deg_x u^i_x = 1$ and we denote by $\mathcal S_d$ the subspace consists of homogeneous elements with differential degree $d$. 

Due to the fact that the deformed bihamiltonian structure is tau-symmetric, it follows from Theorem 3.8.24 of \cite{dubrovin2018bihamiltonian} that the functions $Q_g^\qa$ in the quasi-Miura transformation in \eqref{AQ} can be represented in the form
\begin{equation}
\label{AP}
Q_g^\qa = \eta^{\qa\qb}\frac{\qp^2}{\qp t^{\qb,0}\qp t^{1,0}}T_g(v,v_x,\cdots,v^{(m_g)}),\quad g\geq 2.
\end{equation}
To prove the main theorem of the present paper, we first need to show that $T_g\in\mathcal S_{2g-2}$ and $m_g\leq 3g-2$.

Let us denote by $\mathcal C$ the space of smooth functions depending on the jet variables $u,u_x,u_{xx},\cdots$. It is clear that $\mathcal S$ is a subspace of $\mathcal C$. The derivation $\qp_x$ can be viewed as a linear operator on $\mathcal C$ by
\[
\qp_x = \sum_{i=1}^n\sum_{s\geq 0}u^{i,s+1}\diff{}{u^{i,s}}.
\]
Consider the following operators $\qd_i$ of variational derivatives on $\mathcal C$ defined by
\[
\qd_i = \sum_{s\geq 0}(-\qp_x)^s\diff{}{u^{i,s}},\quad i=1,\cdots,n.
\]
We have the following lemma.
\begin{Lem}
\label{AS}
The following properties for the operators $\qd_i$ hold true:
\begin{enumerate}
\item For any function $f\in\mathcal C$, $\qd_i\kk{\qp_xf}= 0$.
\item If an element $f\in\mathcal S_{\geq 2}$ satisfies $\qd_if = 0$, then there exists $g\in\mathcal S$ such that $f = \qp_x g$.
\end{enumerate}
\end{Lem}
\begin{proof}
The first property is trivial, so we only need to prove the second one. Let us denote by $\mathcal S^{(N)}$ the subspace of $\mathcal S$ consists of elements that do not depend on $u^{i,k}$ for any $i = 1,\cdots,n$ and $k>N$. Without loss of generality, we can assume that $f\in\mathcal S_d$ for $d\geq 2$ and that $f\in\mathcal S^{(N)}$ for a certain integer $N$. 

It follows from the definition of the operator $\qd_i$ that
\[
\diff{}{u^{j,2N}}(\qd_if) = (-1)^N\diff{}{u^{j,N}}
\diff{f}{u^{i,N}}.
\]
So by using $\qd_i f = 0$ we see that $f$ must have the form
\begin{equation}
\label{AO}
f =\sum_{i=1}^nf_i u^{i,N}+g,\quad g,f_1,\cdots,f_n\in\mathcal S^{(N-1)}.
\end{equation}
From the above expression of the function $f$, we arrive at the following identity
\begin{equation*}
0 = \diff{}{u^{j,2N-1}}(\qd_if) = (-1)^N\kk{\diff{f_i}{u^{j,N-1}}-\diff{f_j}{u^{i,N-1}}}.
\end{equation*}
Therefore there exists $h\in\mathcal C$ such that 
\[
f_i = \diff{h}{u^{i,N-1}},\quad i = 1,\cdots, n.
\]
If $N\neq 2$, it is clear that $h$ can be chosen such that $h\in\mathcal S^{(N-1)}$. If $N=2$, we can show that $h$ can be chosen to have the form
\[
h = \sum_{i=1}^n a_i\log u^i_x+b,\quad a_i\in C^\infty(u),\ b\in\mathcal S^{(1)}.
\]
However by using the assumption $f\in\mathcal S_{\geq 2}$, it follows from \eqref{AO} that $f_i\in\mathcal S_{\geq 0}$ and therefore all smooth functions $a_i$ must vanish. We conclude that we can always assume $h\in\mathcal S^{(N-1)}$ and consequently $f-\qp_xh\in\mathcal S^{(N-1)}$. Now the lemma can be proved by induction on $N$.
\end{proof}
\begin{Prop}
\label{AR}
The functions $T_g$ given in \eqref{AP} lie in the space $\mathcal S_{2g-2}^{(3g-2)}$ for $g\geq 2$.
\end{Prop}
\begin{proof}
It follows from \eqref{AP} that $\eta_{1\qa}Q^\qa_g =\qp_x^2T_g$ for $g\geq 2$. Therefore by using Lemma \ref{AS} we know that $\qp_xT_g\in \mathcal S_{2g-1}^{(3g-1)}$. By using this lemma again, we arrive at the fact that $T_g\in \mathcal S_{2g-2}^{(3g-2)}$. The proposition is proved.
\end{proof}

\subsection{Proof of the main theorem}
Let us fix a tau-symmetric deformation of the bihamiltonian structure \eqref{AM},\eqref{AN} of the Principal Hierarchy associated with $M$. We denote by $Z$ the tau-function of the deformed integrable hierarchy determined by the deformation of the bihamiltonian structure and by $w^\qa$ the normal coordinates of the deformed integrable hierarchy. Then the tau-cover of the deformed integrable hierarchy can be written in the form
\begin{equation}
\label{AY}
\qe\diff{\mathcal F}{t^{\qa,p}} = f_{\qa,p},\quad \qe\diff{f_{\qb,q}}{t^{\qa,p}} = \Qo_{\qa,p;\qb,q},\quad \diff{w^\qb}{t^{\qa,p}} = \eta^{\qb\ql}\qp_x\Qo_{\ql,0;\qa,p}.
\end{equation}
According to Theorem \ref{AT}, the deformed integrable hierarchy admits Virasoro symmetries $\diff{}{s_m}$ whose actions on $Z$ are given by
\begin{equation}
\label{AU}
\diff{Z}{s_m} = L_mZ+O_m Z,\quad m\geq -1,
\end{equation}
where $O_m$ are differential polynomials. It is proved in \cite{dubrovin2018bihamiltonian} that after applying a Miura type transformation
\[
w^\qa\mapsto\tilde w^\qa = w^\qa+\sum_{g\geq 1}\qe^{2g}A^\qa_g,
\]
where $A^\qa_g$ are differential polynomials in $w^1,\cdots, w^n$ of differential degree $2g$, the tau-cover of the deformed integrable hierarchy changes accordingly:
\begin{equation}
\label{AV}
\mathcal{\tilde F} = \mathcal F+G,\quad  \tilde{f}_{\qa,p} = f_{\qa,p}+\qe\diff{G}{t^{\qa,p}},\quad \tilde{w}^\qa = w^\qa +\qe^2\eta^{\qa\qb}\frac{\qp^2 G}{\qp t^{\qb,0}\qp t^{1,0}},
\end{equation}
where $G = \sum_{g\geq 1}\qe^{2g-2}G_g$, and $G_g$ are homogeneous differential polynomials of differential degree $2g-2$. Here $\mathcal{\tilde F}$, $\tilde{f}_{\qa,p}$ and $\tilde{w}^\qa$ are the free energy, one-point functions and normal coordinates of the tau-cover of the deformed integrable hierarchy after applying the Miura type transformation. Conversely, any $G = \sum_{g\geq 1}\qe^{2g-2}G_g$ defines a Miura type transformation and changes the tau-cover of the deformed integrable hierarchy according to \eqref{AV}. Let us first study how the Miura type transformations affect the Virasoro symmetries \eqref{AU}.

After applying a Miura type transformation given by \eqref{AV}, the tau-function $Z$ is changed to $\tilde Z = Z\exp(G)$. Then the actions of the Virasoro symmetries \eqref{AU} on $\tilde Z$ are represented in the form
\[
\diff{\tilde Z}{s_m} = L_m\tilde Z+\tilde O_m \tilde Z,\quad m\geq -1,
\]
where $\tilde O_m$ is given by
\begin{align*}
\tilde O_m =&\, O_m-\sum a_{m}^{\qa,p;\qb,q}\kk{2\qe f_{\qa,p}\diff{G}{t^{\qb,q}}+\qe^2\diff{G}{t^{\qa,p}}\diff{G}{t^{\qb,q}}+\qe^2\frac{\qp^2G}{\qp t^{\qa,p}\qp t^{\qb,q}}}\\
&+\diff{G}{s_m}-\sum {b}_{m;\qa,p}^{\qb,q} t^{\qa,p}\frac{\qp G}{\qp t^{\qb,q}},\quad m\geq -1.
\end{align*}
It follows from \eqref{AY} and \eqref{AU} that for any differential polynomial $Q$ in $w^1,\cdots,w^n$, the actions of the Virasoro symmetries on $Q$ are given by
\[
\diff{Q}{s_m} = 2\qe\sum a_{m}^{\qa,p;\qb,q}f_{\qa,p}\diff{Q}{t^{\qb,q}}+X_m+\sum {b}_{m;\qa,p}^{\qb,q} t^{\qa,p}\frac{\qp Q}{\qp t^{\qb,q}},\quad m\geq -1,
\]
where $X_m$ is a certain differential polynomial determined by $L_m$ and $Q$. For this reason it is easy to see that $\tilde O_m$ is also a differential polynomial. Hence to linearize the Virasoro symmetries is equivalent to solve the following equations for a differential polynomial $G$:
\begin{align}
\label{AW}
\diff{G}{s_m} =&\, \sum a_{m}^{\qa,p;\qb,q}\kk{2\qe f_{\qa,p}\diff{G}{t^{\qb,q}}+\qe^2\diff{G}{t^{\qa,p}}\diff{G}{t^{\qb,q}}+\qe^2\frac{\qp^2G}{\qp t^{\qa,p}\qp t^{\qb,q}}}\\
\notag
&+\sum {b}_{m;\qa,p}^{\qb,q} t^{\qa,p}\frac{\qp G}{\qp t^{\qb,q}}-O_m,\quad m\geq -1.
\end{align}

Let us make the expansion 
\[G = \sum_{g\geq 1}\qe^{2g-2}G_g,\quad O_m = \sum_{g\geq 1}\qe^{2g-2}O_{m;g},
\]
where $G_g$ and $O_{m;g}$ are differential polynomials of differential degree $2g-2$. Then by comparing the leading terms of the equations \eqref{AW}, we obtain the following equations for $G_1$:
\[
\mathcal D_m(G_1) = O_{m,1},\quad m\geq -1,
\]
where $\mathcal D_m$ is the operator defined in \eqref{AX}. Once such a differential polynomial $G_1$ is found, we can apply the Miura type transformation
\begin{equation*}
\mathcal{ F} \mapsto \mathcal F+G_1,\quad  {f}_{\qa,p} \mapsto f_{\qa,p}+\qe\diff{G_1}{t^{\qa,p}},\quad {w}^\qa \mapsto w^\qa +\qe^2\eta^{\qa\qb}\frac{\qp^2 G_1}{\qp t^{\qb,0}\qp t^{1,0}}
\end{equation*}
to the tau-cover of the deformation of the Principal Hierarchy. If we still denote by \eqref{AU} the Virasoro symmetries of the deformed integrable hierarchy after the above Miura type transformation, then we see that $O_m = \sum_{g\geq 2}\qe^{2g-2}O_{m;g}$. In order to eliminate $O_{m;2}$ from the representations of the Virasoro symmetries, we need to find a differential polynomial $G_2$ such that
\[
\mathcal D_m(G_2) = O_{m,2},\quad m\geq -1
\]
and perform a Miura type transformation
\begin{equation*}
\mathcal{ F} \mapsto \mathcal F+\qe^2G_2,\quad  {f}_{\qa,p} \mapsto f_{\qa,p}+\qe^3\diff{G_2}{t^{\qa,p}},\quad {w}^\qa \mapsto w^\qa +\qe^4\eta^{\qa\qb}\frac{\qp^2 G_2}{\qp t^{\qb,0}\qp t^{1,0}}.
\end{equation*}
By continuing this procedure recursively, we are able to prove the main theorem. Before turning to the actual proof, we first make some preparations. 

In what follows, we will use $u^i = u^i(w)$ to denote the canonical coordinates of the semisimple Frobenius manifold $M$. Let us recall some basic facts related to the canonical coordinates following \cite{dubrovin2001normal,dubrovin1996geometry}. In terms of $u^1,\cdots,u^n$, the flat metric $\eta$ of $M$ has the expression
\[
\eta = \sum_{i=1}^nf_i(u)(du^i)^2,\quad f_i(u)\neq 0.
\]
We denote by $\psi_{i1} = \sqrt{f_i}$ the Lam\'e coefficients, and by $\qg_{ij} = (\psi_{i1})^{-1}\qp_i\psi_{j1}$ the rotation coefficients for $i\neq j$. Following \cite{dubrovin2001normal} we say that $M$ is \textit{reducible} if there exists a partition of the set $\{1,2,\cdots,n\}$ into two nonempty disjoint subsets $I,J$ such that
\[
\qg_{ij} = 0,\quad \forall\,i\in I,\forall\, j\in J,
\]
and say that $M$ is \textit{irreducible} if it is not reducible. Note that if the dimension $n$ of $M$ satisfies $n\geq 2$, the irreducibility of $M$ implies that for each $i$ there exists $j\neq i$ such that $\qg_{ij}\neq 0$.

We define a strict partial order $\prec$ on the space of homogeneous differential polynomials $\mathcal A_d$ with differential degree $d\geq 1$. For any $W\in\mathcal A_d$, we can represent it in the form
\[
W = \sum W_{\mu_1,\cdots,\mu_n}(u)u^{1,(\mu_1)}\cdots u^{n,(\mu_n)},
\]
where each $\mu_i = (\mu_{i,1},\mu_{i,2},\cdots)$ is a partition of a non-negative integer $|\mu_i|$ and 
\[
u^{i,(\mu_i)} = u^{i,\mu_{i,1}}u^{i,\mu_{i,2}}\cdots.
\]
Note that $|\mu_1|+\cdots+|\mu_n| = d$, hence after reordering $(\mu_1,\cdots,\mu_n)$ we obtain a partition of $d$. Then the relation $\prec$ on $\mathcal A_d$ is induced by the total order on the space of all partitions of $d$. To illustrate this definition, we consider the partial order on $\mathcal A_4$ for a $2$-dimensional Frobenius manifold. For example, we have the relations
\[
(u^{2,1})^3u^{1,1}\prec (u^{1,1})^2u^{2,2}\prec u^{1,3}u^{2,1}\prec u^{2,4},
\]
which correspond to the following relations of the partitions of $4$:
\[
(1,1,1,1)\prec (2,1,1)\prec(3,1)\prec (4).
\]
There are also some monomials that are not comparable with respect to the above partial order, for example $u^{1,3}u^{2,1}$ and $u^{1,1}u^{2,3}$ are not comparable.
\begin{Lem}
\label{BD}
For $k\geq 1$, we have
\[
\mathcal D(\ql)u^{i,k} = -\kk{1+\frac k2}\frac{u^{i,k}}{(u^i-\ql)^2}-ku^{i,k}\sum_{j\neq i}\frac{\psi_{j1}}{\psi_{i1}}\qg_{ij}\kk{\frac{1}{u^j-\ql}-\frac{1}{u^i-\ql}}+R_{i,k},
\]
here $R_{i,k}$ consists of terms smaller than $u^{i,k}$ with respect to the order $\prec$. In particular, $R_{i,1} = 0$.
\end{Lem}
\begin{proof}
In terms of the canonical coordinates, we can write $\mathcal D(\ql)$ in the form
\[
\mathcal D(\ql) = \sum_{i=1}^n\sum_{r\geq 0}\qp_x^r\kk{\frac{1}{u^i-\ql}}\diff{}{u^{i,r}}+\sum_{i=1}^n\sum_{r\geq 1}B_{i,r}\diff{}{u^{i,r}}.
\] 
It is proved in \cite{dubrovin2001normal} that the functions $B_{i,r}$ satisfy the relations
\begin{align*}
B_{i,1} &= -\frac 12\frac{u^{i,1}}{(u^i-\ql)^2}-u^{i,1}\sum_{j\neq i}\frac{\psi_{j1}}{\psi_{i1}}\qg_{ij}\kk{\frac{1}{u^j-\ql}-\frac{1}{u^i-\ql}}\\ &= \diff{u^i}{v^\qz}\eta^{\qz\ql}\qp_\rho\qp_\ql p_\qa G^{\qa\qb}\diff{p_\qb}{v^1}v^{\rho,1},\\
B_{i,r} &= \qp_xB_{i,r-1}+\qp_x^{r-1}\kk{\diff{u^i}{v^\qz}\eta^{\qz\ql}\qp_\rho\qp_\ql p_\qa v^{\rho,1}}G^{\qa\qb}\diff{p_\qb}{v^1},\quad r\geq 2.
\end{align*}
The lemma is proved by a straightforward computation.
\end{proof}
\begin{Lem}[ Lemma 4.2.2 of \cite{dubrovin2001normal}]
\label{BE}
Let $F\in\mathcal S^{(N)}$, then $\mathcal D(\ql)F$ is a rational function in $\ql$ with poles at most of order $N+1$ at each $\ql = u^i$ and it is regular at $\ql=\infty$. The coefficients of $(\ql-u^i)^{-N-1}$ are given by $-C_N(u^{i,1})^N\diff{F}{u^{i,N}}$ where $C_N$ are positive numbers given by
\[
C_N = N!+2^{-N}\sum_{l=1}^N\binom{N}{l}(2l-3)!!(2N-2l+1)!!.
\]
\end{Lem}
\begin{Prop}
\label{BF}
Let $M$ be a semisimple irreducible Frobenius manifold of dimension $n\geq 2$, and $F\in \mathcal S_d$ with $d\geq 1$ satisfy the identity
\[
\mathcal D(\ql)F = H
\]
for a certain $H\in\mathcal A_d$, then we have $F\in\mathcal A_d$. 
\end{Prop}
\begin{proof}
Let us assume that $F\in\mathcal S^{(N)}$, $H\in\mathcal A^{(N)}$ for a certain positive integer $N$, here $\mathcal A^{(N)} = \mathcal A\cap \mathcal S^{(N)}$.
Then we can represent $F$ in the following form:
\[
F = \sum_{(k_1,\cdots,k_n)\neq 0}W_{k_1,\cdots,k_n}(u^{1,N})^{k_1}\cdots(u^{n,N})^{k_n}+R,\quad W_{k_1,\cdots,k_n},R\in\mathcal S^{(N-1)}. 
\]
By using Lemma \ref{BE} we know that the coefficients of $(\ql-u^i)^{-N-1}$ in $\mathcal D(\ql)F$ are given by
\begin{equation}
\label{BA}
\sum_{(k_1,\cdots,k_n)\neq 0}-k_iC_N(u^{i,1})^NW_{k_1,\cdots,k_n}(u^{1,N})^{k_1}\cdots(u^{i,N})^{k_i-1}\cdots(u^{n,N})^{k_n}.
\end{equation}
If $N=1$, then for each term $W_{k_1,\cdots,k_n}(u)(u^{1,1})^{k_1}\cdots(u^{n,1})^{k_n}$ with $k_i\neq 0$ for some $i$, we must have $k_i>0$ due to the assumption that $\mathcal D(\ql)F$ is a differential polynomial. Therefore we conclude that  $F\in\mathcal A$, hence in what follows we assume $N\geq 2$. 

It follows from \eqref{BA} that 
\[k_iC_N(u^{i,1})^NW_{k_1,\cdots,k_n}\in\mathcal A,\quad i=1,\cdots, n.\] 
So it is easy to see that for a term $W_{k_1,\cdots,k_n}(u^{1,N})^{k_1}\cdots(u^{n,N})^{k_n}$, if there exists $i\neq j$ such that $k_i\neq 0$ and $k_j\neq 0$, then we must have $W_{k_1,\cdots,k_n}\in\mathcal A$. In another word, we can uniquely represent $F$ in the following form:
\[
F = R+P+\sum_{i=1}^n\sum_{r}V^i_r\frac{(u^{i,N})^{k_{i,r}}}{(u^{i,1})^{l_{i,r}}},\quad  R\in\mathcal S^{(N-1)}, P\in\mathcal A^{(N)},V^i_r\in\mathcal A^{(N-1)},
\]
here $V^i_r$ do not depend on $u^{i,1}$ and $l_{i,r},k_{i,r}>0$. Let us proceed to show that
\[
\diff{V^i_r}{u^{j,s}} = 0,\quad j\neq i,\ s\geq 0.
\]
For each $j\neq i$, we denote by $s$ the maximal integer such that $\diff{V^i_r}{u^{j,s}} \neq 0$. Let us consider the coefficients of
\begin{equation}
\label{BB}
\frac{1}{(\ql-u^j)^{s+1}}\frac{(u^{i,N})^{k_{i,r}}}{(u^{i,1})^{l_{i,r}}}
\end{equation}
in the expression of $\mathcal D(\ql)F$. It is clear that both $R$ and $P$ give no contributions to this term and we see, by using Lemma \ref{BE} again, that the coefficient of \eqref{BB} is given by
\[
-\diff{V^i_r}{u^{j,s}}C_s(u^{j,1})^s.
\]
This contradicts the fact that $\mathcal D(\ql)F$ is a differential polynomial and therefore $\diff{V^i_r}{u^{j,s}} = 0$ for any $j\neq i$ and $s\geq 0$. So we can uniquely represent $F$ in the form
\[
F = R+P+\sum_{i=1}^n\sum_{l_1;k_2,\cdots,k_N} G^i_{k_1,\cdots,k_N}(u^i)\frac{(u^{i,2})^{k_2}\cdots(u^{i,N})^{k_N}}{(u^{i,1})^{l_{1}}},\quad l_1,k_N> 0.
\]
For each index $i$, let us pick up a nonzero term
\[
G^i_{k_1,\cdots,k_N}(u^i)\frac{(u^{i,2})^{k_2}\cdots(u^{i,N})^{k_N}}{(u^{i,1})^{l_{1}}}
\]
such that $l_1$ is maximal and its numerator is maximal with respect to $\prec$. Then we can apply Lemma \ref{BD} and compute the coefficients of
\[
\frac{1}{\ql-u^j}\frac{(u^{i,2})^{k_2}\cdots(u^{i,N})^{k_N}}{(u^{i,1})^{l_{1}}},\quad j\neq i
\]
in the expression of $\mathcal D(\ql)F$, which are given by
\[
G^i_{k_1,\cdots,k_N}(2k_2+3k_3\cdots+Nk_N-l_1)\frac{\psi_{j1}}{\psi_{i1}}\qg_{ij},\quad j\neq i.
\]
Due to the assumption that $F\in\mathcal S_d$, we see that $2k_2+\cdots+Nk_N-l_1 = d\geq 1$ and it follows from the irreducibility of $M$ that we can choose $j\neq i$ such that $\qg_{ij}\neq 0$. Therefore by using $\mathcal D(\ql)F\in\mathcal A$, we conclude that $G^i_{k_1,\cdots,k_N} = 0$, which contradicts to our choice. Hence we conclude that $F$ must have the form
\[
F = R+P,\quad R\in\mathcal S^{(N-1)},P\in\mathcal A.
\]
We then arrive at
\[
\mathcal D(\ql)(R) = H-\mathcal D(\ql)(P)\in\mathcal A,\quad R\in\mathcal S^{(N-1)}.
\]
Then we can finish the proof of the proposition by doing induction on $N$.
\end{proof}

Now we are ready to prove the main theorem.
\begin{proof}[\textbf{Proof of the Main Theorem}] The `only if' part is proved by combining the results of \cite{dubrovin1998bihamiltonian,dubrovin2001normal,ICCM2007}. In fact, due to the linearization of the Virasoro symmetries, it follows from Theorem \ref{AL} that the quasi-Miura transformation \eqref{AQ} must be given by
\[
Q_g^\qa = \eta^{\qa\qb}\frac{\qp^2F_g}{\qp t^{\qb,0}\qp t^{1,0}},\quad g\geq 1,
\]
where $F_g$ are described as in Theorem \ref{AL}. Thus the explicit formula of the genus one component of the deformed bihamiltonian structure is given by \cite{dubrovin1998bihamiltonian} and its central invariants are all equal to $\frac{1}{24}$ \cite{ICCM2007}.

Let us turn to prove the `if' part. If the dimension $n$ of the semisimple Frobenius manifold $M$ is equal to $1$, then it is well known that the corresponding Principal Hierarchy is the dispersionless KdV hierarchy. It admits a deformation given by the deformed bihamiltonian structure
\begin{align*}
&\{w(x),w(y)\}_1 = \qd'(x-y),\\ &\{w(x),w(y)\}_2 = w(x)\qd'(x-y)+\frac 12w'(x)\qd(x-y)+\frac{\qe^2}{8}\qd'''(x-y).
\end{align*}
This particular deformation determines an integrable hierarchy with linearized Virasoro symmetries \cite{dubrovin2001normal,kontsevich1992intersection,witten1990two}. Hence in what follows we assume the dimension $n$ of the semisimple Frobenius manifold $M$ satisfies $n\geq 2$ and without loss of generality, we assume $M$ is irreducible.

Recall that the central invariants completely determine the equivalent class of  deformations of a semisimple bihamiltonian structure of hydrodynamic type under Miura type transformation. Therefore if the central invariants of the deformed bihamiltonian structure are all equal to $\frac{1}{24}$, then there exists a Miura type transformation such that, after applying this transformation, the genus one component of the deformed bihamiltonian structure coincides with the one given in \cite{dubrovin1998bihamiltonian}. This Miura type transformation cancels $O_{m,1}$ due to the result of \cite{dubrovin2001normal,dubrovin1999frobenius}. Now let us do indunction and assume that we have canceled $O_{m,k}$ for all $m\geq -1$ and $k\leq g$ for a certain $g\geq 1$ in the representations of the Virasoro symmetries \eqref{AU}. Then we are going to solve $G_{g+1}$ such that
\begin{equation}
\label{AZ}
\mathcal D(\ql)G_{g+1} = O_{g+1}(\ql),\quad O_{g+1}(\ql) = \sum_{\ql\geq -1}\frac{1}{\ql^{m+2}}O_{m;g+1}.
\end{equation}
It follows from Proposition \ref{AR} that the quasi-Miura transformation of the deformed bihamiltonian structure can be represented by
\[
v^\qa\mapsto w^\qa = v^\qa+\qe^2 Q^\qa_1+\sum_{k\geq 2}\qe^{2k}\eta^{\qa\qb}\frac{\qp^2}{\qp t^{\qb,0}\qp t^{1,0}}T_k\bigl(v,v_x,\cdots,v^{(3g-2)}\bigr),\quad T_k\in \mathcal S_{2k-2}^{(3k-2)}.
\]
Due to the induction hypothesis and the uniqueness of the quasi-Miura transformation, we see that 
\[
Q_1^\qa = \eta^{\qa\qb}\frac{\qp^2  F_1^{DZ}}{\qp t^{\qb,0}\qp t^{1,0}},\quad T_k =  F_k^{DZ},\quad k\leq g,
\]
here we denote by $ F^{DZ}_k$ the solution of the loop equation \eqref{AI} given by Theorem \ref{AL}. 

On the other hand, from the tau-structure \eqref{AY} of the deformed integrable hierarchy it follows that the genus expansion of the free energy $\mathcal F = \qe^{-2}\mathcal F_0+\sum_{k\geq 1}\qe^{2k-2} F_k$ satisfies the identities
\[
 F_1 =  F_1^{DZ},\quad  F_k = T_k,\quad k\geq 2.
\]
Similar to the the derivation of the loop equation given in \cite{dubrovin2001normal}, we obtain from the expressions of the  Virasoro symmetries
\[
\diff{Z}{s_m} = L_mZ+O_mZ
\] 
and the genus zero Virasoro symmetries \eqref{AC} that the function $ F_{g+1}$ satisfies the equation
\begin{align*}
\mathcal D(\ql) F_{g+1}
=&\,\frac{1}{2}\kk{\sum_{k=1}^{g-1}\diff{ F_k}{v^{\qg,k}}\diff{ F_{g-k}}{v^{\rho,l}}+\frac{\qp^2 F_{g}}{\qp v^{\qg,k}\qp v^{\rho,l}}}\qp_x^{k+1}\qp^\qg p_\qa G^{\qa\qb}\qp_x^{l+1}\qp^\rho p_\qb\\
&+\frac 12\sum\diff{ F_g}{v^{\qg,k}}\qp_x^{k+1}\left[\nabla\diff{p_\qa}{\ql}\cdot \nabla\diff{p_\qb}{\ql}\cdot v_x\right]^\qg G^{\qa\qb}-O_{g+1}(\ql).
\end{align*}
Recall that the loop equation \eqref{AI} satisfied by $ F_{g+1}^{DZ}$ is given by
\begin{align*}
\mathcal D(\ql) F_{g+1}^{DZ}
=&\,\frac{1}{2}\kk{\sum_{k=1}^{g-1}\diff{ F_k^{DZ}}{v^{\qg,k}}\diff{ F_{g-k}^{DZ}}{v^{\rho,l}}+\frac{\qp^2 F_{g}^{DZ}}{\qp v^{\qg,k}\qp v^{\rho,l}}}\qp_x^{k+1}\qp^\qg p_\qa G^{\qa\qb}\qp_x^{l+1}\qp^\rho p_\qb\\
&+\frac 12\sum\diff{ F_g^{DZ}}{v^{\qg,k}}\qp_x^{k+1}\left[\nabla\diff{p_\qa}{\ql}\cdot \nabla\diff{p_\qb}{\ql}\cdot v_x\right]^\qg G^{\qa\qb}.
\end{align*}
Therefore it follows from the induction hypothesis that equation \eqref{AZ} has a solution
\[G_{g+1} =  F_{g+1}^{DZ}- F_{g+1}\in\mathcal S_{2g}\] 
such that
\[
\mathcal D(\ql)G_{g+1} = O_{g+1}.
\]
Note that $O_{g+1}(\ql)\in\mathcal A_{2g}$ and $g\geq 1$, then by applying Proposition \ref{BF} we see that $G_{g+1}$ is also a differential polynomial. The theorem is proved.
\end{proof}

\section{Conclusion}\label{conc}


In \cite{dubrovin2001normal}, a program for classifying the integrable hierarchies of topological type was proposed, the main goal of which is to describe a 2DTFT completely via an integrable hierarchy. These integrable hierarchies should satisfy the following four axioms:
\begin{description}
\item[Axiom BH] The integrable hierarchy possesses a bihamiltonian structure with dispsersionless limit of hydrodynamic type.

\item[Axiom QT] The integrable hierarchy is quasi-trivial, meaning that there is a quasi-Miura transformation that changes the integrable hierarchy to its dispersionless limit.

\item[Axiom TS]  The integrable hierarchy possesses a tau structure.

\item[Axiom LVS] The integrable hierarchy possesses Virasoro symmetries which are given by linear actions of Virasoro operators on the tau function.
\end{description}
Such an integrable hierarchy is called of topological type. The central problem regarding this classification program   is the existence and uniqueness of such kinds of integrable hierarchies. In \cite{dubrovin2001normal}, it is proved that starting from the axiom BH and axiom TS, the dispersionless limit of such an integrable hierarchy corresponds to a Frobenius manifold or a degenerate one. On the other hand, for a given caliberated semisimple Frobenius manifold, there exists a unique deformation of its Principal Hierarchy satisfies all four axioms, which is just the DZ hierarchy as it is constructed in \cite{dubrovin2001normal}. The remaining work to be finished is to prove that the bihamiltonian structure of DZ hierarchy can be represented in terms of differential polynomials. In \cite{buryak2012deformations,buryak2012polynomial}, Buryak, Posthuma and Shadrin proved that the first Hamiltonian structure of the DZ hierarchy can be represented by differential polynomials by considering the infinitesimal transformation of Givental's group action \cite{givental2001gromov,lee2009invariance}. In the present paper we prove the polynomiality property of the bihamiltonian structure of the DZ hierarchy by using the super tau-covers of integrable hierarchies\cite{liu2020super} and the associated variational bihamiltonian cohomologies\cite{liu2021variational,variationalII}. 

Our strategy is to study the relations among these four axioms. The relations among the first three were studied in the papers \cite{DLZ-1,dubrovin2018bihamiltonian,falqui2012exact,liu2005deformations}. Under the semisimplicity condition, axiom QT can be derived from axiom BH, and axiom TS can also be derived form axiom BH by requiring that the bihamiltonian structure is flat exact\cite{dubrovin2018bihamiltonian}. The present paper together with the work \cite{variationalII,liu2020super,liu2021variational} gives a clear description of the relation between axiom BH and axiom LVS, that is, the axiom LVS can also be derived from axiom BH by requiring that the central invariants of the flat exact bihamiltonian structure are all equal to $\frac{1}{24}$. To conclude, we have the following theorem.
\begin{Th}
For a given semisimple Frobenius manifold, a bihamiltonian deformation of its Principal Hierarchy is equivalent to the DZ hierarchy under Miura type transformations if and only if the central invariants of the deformed bihamiltonian structure are all equal to $\frac{1}{24}$.
\end{Th}

A possible application of this theorem is to study the so-called DR/DZ equivalence conjecture in the semisimple setting. In \cite{buryak2015double}, Buryak constructed an integrable hierarchy, called the double ramification (DR) hierarchy, starting from any cohomological field theory (CohFT). He showed that the genus zero part of the DR hierarchy coincides with the Principal Hierarchy of the Frobenius manifold given by the genus zero part of the CohFT. It was conjectured that the DR hierarchy is equivalent to the DZ hierarchy for any semisimple CohFT under Miura type transformations, and this is called the DR/DZ equivalence conjecture. Although the conjecture is still open, there have been many evidences and verifications for this conjecture. In the paper \cite{buryak2021towards}, a conjectural formula was given for the bihamiltonian structure of the DR hierarchy and under the assumption of the correctness of their conjectural formula, the central invariants were computed and they are all equal to $\frac{1}{24}$. In the papers \cite{ brauer2021bihamiltonian,buryak2018tau}, the DR/DZ equivalence conjecture was proved up to genus one approximation and in particular the conjectural formula for the bihamiltonian structure of the DR hierarchy in \cite{buryak2021towards} is correct up to genus one. Therefore by applying the above theorem, the DR/DZ equivalence conjecture can be proved in the semisimple setting by confirming the correctness of the conjectural formula for the bihamiltonian structure of the DR hierarchy.


\medskip

\noindent Si-Qi Liu,

\noindent Department of Mathematical Sciences, Tsinghua University \\ 
Beijing 100084, P.R.~China\\
liusq@tsinghua.edu.cn
\medskip

\noindent Zhe Wang,

\noindent Department of Mathematical Sciences, Tsinghua University \\ 
Beijing 100084, P.R.~China\\
zhe-wang17@mails.tsinghua.edu.cn
\medskip

\noindent Youjin Zhang,

\noindent Department of Mathematical Sciences, Tsinghua University \\ 
Beijing 100084, P.R.~China\\
youjin@tsinghua.edu.cn

\end{document}